\documentclass[runningheads]{llncs}

\newtheorem{algorithm}{Algorithm}
\usepackage{tikz}
\usepackage{hyperref}
\usepackage{istgame}
\usetikzlibrary{arrows.meta}
\usetikzlibrary{calc,shapes.callouts,shapes.arrows}
\usepackage{colortbl}
\usepackage{amsmath} 
\usepackage{multirow, multicol, array, makecell}
\usepackage[capitalise]{cleveref}

\usepackage[
backend=biber,
style=numeric,
]{biblatex}
\begin{document}
\title{Incentive Non-Compatibility of Optimistic Rollups}
%
%
\author{Daji Landis\inst{1}}
\authorrunning{D. Landis}
%
\institute{New York University} 
%
\maketitle              
\begin{abstract}
Optimistic rollups are a popular and promising method of increasing the throughput capacity of their underlying chain. These methods rely on economic incentives to guarantee their security.  We present a model of optimistic rollups that shows that the incentives are not aligned with the expected behavior of the players, thus potentially undermining the security of existing optimistic rollups.  We discuss some potential solutions illuminated by our model.
\begin{keywords}
Rollup, Blockchain, Scaling 
\end{keywords}
\end{abstract}
\section{Introduction}
Blockchains are public ledgers used to record transactions and execute decentralized programs. Rollups are a `Layer 2' (L2) scaling solution that operate on top of a main blockchain \cite{scaling_survey}.  They add throughput capacity by moving the computations necessary for transactions off the mainchain on to the L2 chain. The data is still posted on the mainchain; a fact that constitutes the main difference between rollups and other scaling solutions. There are two major types of rollups. Optimistic rollups function by `optimistically' expecting that transactions are being committed to a blockchain accurately and honestly in general \cite{scaling_survey, op}.  The execution of these commitments is overseen by other nodes, but are not necessarily systematically checked.  These differ from zero knowledge rollups, which require every commitment to be proved correct every time.

Optimistic rollups function by having two types of nodes. There are nodes that collect and perform the calculations for the bundle of transactions. They must put up a stake to participate. We will call these nodes \emph{aggregators}, but note that this role is sometimes also referred to as the \emph{sequencer}. The work done by these nodes is then checked by other nodes, \emph{validators} that watch the chain and submit \emph{fraud proofs} when they suspect an inaccuracy. These validators must stake currency before they can press a challenge and lose their stake if this challenge is erroneous. If a challenge is successful, the validator receives some of the aggregator's stake and the rest is burned. If a fraud proof is brought, it must be settled on the mainchain, which can be costly, so the system is intended to disincentivize unnecessary proofs \cite{scaling_survey, op}.  

Currently, the most popular implementations of optimistic rollups are Arbitrum and Optimism, with approximately 65\% and 19\% of the market share, respectively \cite{mourya_ethereum_2023}. They differ in their precise implementation of the fraud proof procedure. Neither implementation has a decentralized aggregator, although they both state in their roadmap that decentralization is a priority and coming soon. Moving forward, we will assume that the aggregator is decentralized, but will remain agnostic on how precisely the fault proof is executed.   

Optimistic rollups require users to wait to withdraw their deposits until a `challenge period' has ended.  This allows time for the validator nodes to check the transactions. This period is usually seven days. There are proposed solutions to this issue, for example the use of liquidity pools to allow users to access funds before their transactions are final, but the waiting period will always be an inherent part of the protocol.  A more pressing and well-known problem is that of the `verifier's dilemma', which observes that the verifier has no incentive to execute its role if there are never any errors and therefore no rewards. In this paper, we present a formal model of the verifier's dilemma and discuss potential solutions that our model illuminates.

\subsection{Related Work}
The verifier's dilemma is a frequently discussed issue, with much of that discussion happening in informal blogs.  One of the aims of this paper is to present a formal model of an already well-established line of informal inquiry.  A concurrent work by Li\cite{li} engages with a similar question.  In the paper, the author presents a model of a potential attack on the well functioning of optimistic rollups and concludes that their current design is not secure. Here, we consider a mechanism to be \emph{secure} only if there is a unique pure equilibrium that is the intended behaviour.   There are some substantial differences between their model and ours.  We assume that the choice on the part of the validator of whether or not to perform the search is distinct from the choice of whether or not to challenge, i.e. a validator can blindly challenge without having performed a search. We also assume different payoffs. If a dishonest transaction is caught, we assume the aggregator does not get the utility of the attack, that is, the faulty commitment is revoked.  So a caught dishonest aggregator has their stake slashed while also losing the utility of the attack.  Furthermore, we assume that an erroneous challenge results in a penalty to the validator.

While Li's model and the model we present are different, the general conclusion, that optimistic rollups are not incentive compatible, is similar.  We believe that our model captures the phenomenon in a more nuanced and accurate way.  The paper by Li \cite{li} also extends their model to a case with multiple of each type of player. A similar analysis with many players was done by researchers at Offchain Labs in another concurrent work.  They also find an undesirable equilibrium when there is only one of each player \cite{mamageishvili}.  Our model alone captures the role of search costs and the ability to randomly challenge.  We show that the desired behaviour is not incentive compatible, even with only one of each player.  This work shows the gravity of the verifier's dilemma in a succinct but accurate model, which constitutes an important tool the community has lacked.

Other works use rational and economic augments to investigate other blockchain designs.  This type of analysis has been applied to Byzantine Fault Tolerance \cite{bft2, citation-key}, proof-of-stake \cite{saleh, eth_three_attacks}, sharding \cite{sharding}, and proof-of-work systems \cite{pow, btc, btc_2}.  A survey paper details diverse applications of game theory to blockchain scaling mechanisms \cite{scaling_survey}. There is extensive research into the improvement and functioning of scaling solutions \cite{committee_scaling}. The work in Thibault et al \cite{scaling_survey} discusses many types of blockchain scaling solutions and we follow their account of the functioning of optimistic rollups in the setup of our model. 

\section{Our Model}
Our model is set up with three players: the aggregator $A$, the validator $V$. The owners of the transactions are not taken into account, but they pay some fee $f$, which goes to the aggregator as compensation. In order to commit the transaction, the aggregator must perform some underlying computational work that incurs cost, $w$. This is a computational cost, not a monetary cost, but we will assume that this and other computational costs are easily transformed into monetary costs. We will assume that the agents enjoy simple linear utility. Thus, in this stylized model of a transaction, we have the following utilities, if all goes well: $ A: f-w$ and $V:0$.  The algorithm functions as follows.
\begin{algorithm} \label{optimistic}
\textsc{\textup{Optimistic Rollup Protocol}}.
\begin{enumerate}
    \item Transactor sends their payment $f$ and their transaction information to the aggregator $A$.
    \item The aggregator stakes $s_A$ 
    \item Then $A$ performs the computation of the transaction, with cost $w$ and commits the result.
    \item Then $V$ monitors the integrity of the transaction:
    \begin{enumerate}
        \renewcommand{\theenumii}{\theenumi.\arabic{enumii}}
        \item If $V$ finds an error, it stakes $s_V$ and the proof is checked.
        \begin{enumerate}
        \renewcommand{\theenumiii}{\theenumii.\arabic{enumiii}}
            \item If the challenge is found to be erroneous, the payment and stake can be withdrawn by $A$ and $V$'s stake is slashed.
            \item If the challenge is found to be correct, half the aggregator stake $s_A$ is given to $V$ and the other half is burned.
        \end{enumerate}
        \item If $V$ does not find an error after seven days, the aggregator is able to withdraw $f$ and $s_A$.
    \end{enumerate}
\end{enumerate}
\end{algorithm}

With this algorithm in hand, we can begin clarifying the incentives for each party.  Initially, we will assume that the validator nodes incur no cost to keep abreast of the activities on chain, but we will introduce this cost in the next section. Thus, we are assuming that $V$ always searches for free, and therefore has perfect information.

One important question is what happens when something goes wrong.  If a lazy aggregator acts dishonestly, they might just abscond with the fee $f$.  Rather than being content with just the fee, the aggregator could siphon currency from a smart contract. A particularly large sum might be siphoned from a bridge contract. One attack, through different means than those detailed here, was able to steal \$326M from a Solana bridge contract \cite{bridge}. We call this potentially large benefit $z$.  Attempting this attack will require about the same computational work $w$. With these assumptions, we can readily write down the extensive form game and find the equilibrium. 

\begin{lemma} \label{thm:optimistic_perf_info}
    The optimistic rollup paradigm as described in Algorithm 1 yields ($A$ behaving honestly) and ($V$ posing no challenge) as the pure-strategy Nash equilibrium when $V$ has perfect information and all players are rational and risk neutral. 
\end{lemma}
\begin{proof} 
We have the game in extensive form as follows, with utility vectors of the form $(A, V)^T$:
\begin{center}
\begin{istgame}[font=\footnotesize]
\xtShowEndPoints
\xtdistance{15mm}{40mm}
\istroot(0)[initial node]{A}
    \istb{Honest}[al]
    \istb{Dishonest}[ar]
\endist

\xtdistance{15mm}{20mm}
\istroot(1)(0-1)<135>{V}
    \istb{Challenge}[al]{
    \begin{array}{c}
         A: \\
         V:
    \end{array}
    \begin{pmatrix}
        f-w, \\ -s_v 
    \end{pmatrix}
    }
    \istb{No}[ar]{\begin{pmatrix}
        f-w, \\ 0 
    \end{pmatrix}}
\endist 

\istroot(2)(0-2)<above right>{V}
    \istb{Challenge}[al]{\begin{pmatrix}
        -s_A -w,\\ \frac12 s_A 
    \end{pmatrix}}
    \istb{No}[ar]{\begin{pmatrix}
        f+z-w ,\\ 0 
    \end{pmatrix}}
\endist
\end{istgame}

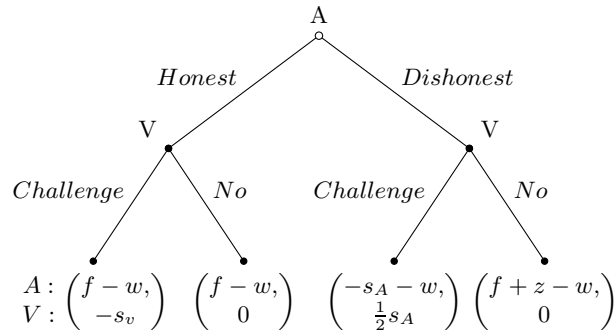
Game 1: Game tree corresponding to Algorithm 1
\end{center}

We implicitly assume that each player also has an outside option wherein they do not partake and receive a utility of $0$. Using backwards induction, we readily see that if $A$ is honest, $V$'s best response is to not pose a challenge.  If $A$ is not honest, $V$'s best response is to pose a challenge.

Given these strategies for $V$, $A$ can see that honest behaviour will result in a utility of $f -w$, which is positive under our assumptions, and dishonest behaviour will result in a utility of $-s_P$.  Thus, $A$'s best strategy is to behave honestly. Thus ($A$ behaving honestly) and ($V$ posing no challenge) is the pure-strategy Nash equilibrium, as desired. \qed
\end{proof}

\section{Search costs}
Even without introducing search costs, it is easy to read off the fact that, given that honest behavior on the part of the prover is rational, $V$ will have an expected utility of $0$. Why then would a validator participate in the protocol? We observe participation in real implementations and there are good reasons to participate other than the potential of the reward. These include the fact that the validator nodes might be run by parties with an interest in the rollup protocol.  Frequent users, for example, might want to keep an eye on the blockchain in order to be assured of the security of their own transactions.  Running a node gives access to the shared ledger, which grants access to data that is important for decentralized applications or frequent traders.  Running a validator node is not as computationally intensive as other types of nodes in other applications. In order to model this accurately, we introduce a small search cost, $x$.  If the validator incurs this search cost, they have perfect information about $A$'s honesty.  

In order to model this addition, we introduce another action where the validator decides whether they intend to invest in the search cost or not. In the following, the dotted lines represent information sets. In context, $A$ will not know if $V$ is performing searches or not, and, if $V$ has chosen (No Search), they do not know if $A$ is honest.  In order to simplify the calculations, we normalize the honest behavior of the aggregator to $0$.  That is set $f-w=0$. We will also simplify the loss that $A$ faces when they misbehave to be just the relatively large slashed stake $s_A$ and not the $w$ term. 

The utility vectors are of the form $(A,V)^T$.  

\centerline{
\overfullrule=0pt 
\begin{istgame}[font=\footnotesize]
\xtShowEndPoints
\xtdistance{15mm}{80mm}
\istroot(0)[initial node]{V}
  \istb{\text{\emph{No Search}}}[al]       
  \istb{Search}[ar]
  \endist
\xtdistance{15mm}{40mm}
\istroot(1)(0-1)<135>{A}
  \istb{Honest}[al]   
  \istb{Dishonest}[ar] 
  \endist
\istroot(2)(0-2)<above right>{A}
  \istb{Honest}[al]
  \istb{Dishonest}[ar]
  \endist
\xtdistance{15mm}{20mm}
\istroot(3)(1-1)<135>{V}
  \istb{Challenge}[al]{
      \begin{array}{c}
         A: \\
         V:
    \end{array}
  \begin{pmatrix}0, \\-s_V \end{pmatrix}} 
  \istb{No}[ar]{ \begin{pmatrix}
      0 ,\\ 0
  \end{pmatrix}}
  \endist
\istroot(4)(1-2)<above right>{V}
  \istb{Challenge}[al]{ \begin{pmatrix}
      -s_A ,\\ \frac12 s_A
  \end{pmatrix} }
  \istb{No}[ar]{\begin{pmatrix}
      z, \\ 0
  \end{pmatrix}}
  \endist
\istroot(5)(2-1)<135>{V}
  \istb{Challenge}[al]{\begin{pmatrix}
      0, \\-s_V-x
  \end{pmatrix}} 
  \istb{No}[ar]{\begin{pmatrix}
      0,\\-x
  \end{pmatrix}}
  \endist
\istroot(6)(2-2)<above right>{V}
  \istb{Challenge}[al]{\begin{pmatrix}
      -s_A, \\ \frac12 s_A-x
  \end{pmatrix}}
  \istb{No}[ar]{\begin{pmatrix}
      z, \\-x
  \end{pmatrix}}
  \endist
\xtInfoset(3)(4)
\xtInfoset(1)(2)

\end{istgame}
}
\begin{center}
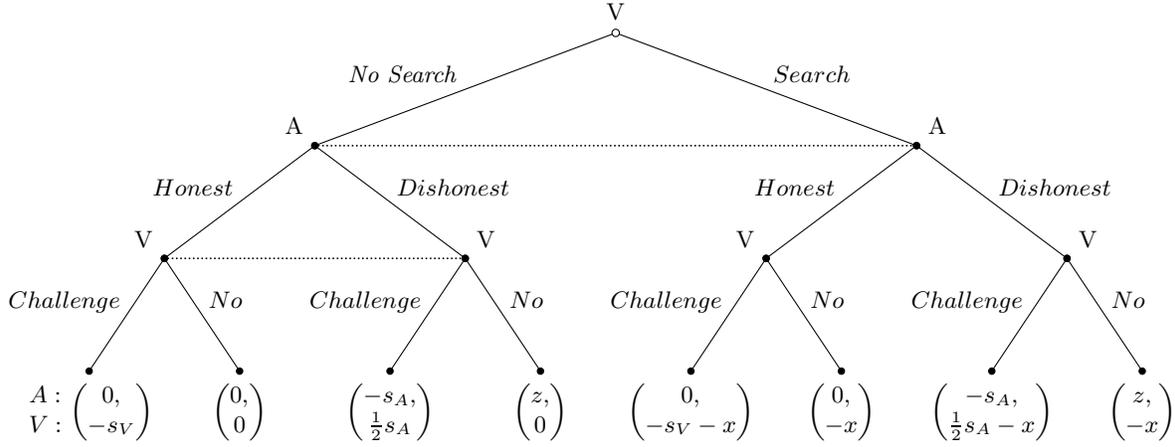
  Game 2: Extensive form game for optimistic rollup with search costs where the dotted lines indicate information sets.  
\end{center}

The right tree corresponds to the game we have already discussed, except that $A$ does not know that $V$ is performing searches. We now assign three probabilities to the choices of the players: $V$ does no search, and makes choices \emph{blindly}, with probability $b$ and then challenges with probability $g$; and $A$ is honest with probability $h$.  All are between $0$ and $1$. Thus we can read off the utility functions, starting with $A$. If the aggregator is honest, they will receive $0$ utility, regardless of $V$'s actions.  Recall this is due to our normalization of $f-w$ to $0$.  When $V$ is blind and is challenging with a probability $g$, dishonesty results in utility $g(-s_A)+(1-g)z$.  When $V$ is not blind, they challenge if and only if $A$ is indeed dishonest, yielding $-s_A$ utility for $A$.  Combining these pieces, we get, 
\begin{align}
    A: \quad h\cdot 0 + (1-h)(b\cdot(g\cdot(-s_A)+(1-g)z)+(1-b)(-s_A)) \label{a_utility}
\end{align}

When $V$ chooses no search, we can see that the validator reaps $0$ utility if they do not challenge. In the challenge case, if $A$ is honest, $V$ has their stake burned for utility $-s_V$.  In the dishonest case, they receive $\frac12 s_A$.  Weighting this with the probabilities, we find the left hand side has expected utility $g\cdot\left(h\cdot(-s_V) + (1-h)\frac12 s_A\right)$.
On the right hand side, $V$ can now make an informed decision.  If $A$ is honest, $V$ will not challenge, resulting in utility $-x$ and if $A$ is dishonest, $V$ will always challenge for utility $\frac12 - x$.  Weighting these terms by the honesty probability yields the right hand side utility of $ h\cdot(-x)+ (1-h)\left(\frac12 s_A -x\right)$, which can be simplified to $(1-h)\left(\frac12 s_A \right)-x$. To get the whole expected utility, we weight both sides of the tree by the probability of no search $b$, which gives,
\begin{align}
    V: \quad bg\cdot\left(h\cdot(-s_V) + (1-h)\frac12 s_A\right)+(1-b)\left((1-h)\frac12 s_A -x\right) \label{v_utility}
\end{align}

\subsection{Main Result}
It suffices to show that there exist values for which the equilibrium is not the desired outcome of full search and honesty. The fully general proof is available in the online version, but the example here more readily illuminates the key interactions.  We start by noting that we should have the following ordinal constraint based on our context: $x<s_V\leq s_A \ll z$.  Take $s_A=s_V=s$, $x=\frac{s}{24}$, and $z=24s$, convenient choices among many that clearly obey these constraints.
\begin{lemma}
    A rational, risk neutral $A$ will be indifferent to being honest if $s_A=s_V=s$ and $z=24s$ and $g= 1- \frac{1}{25b}$.
\end{lemma}
\begin{proof}
    From \cref{a_utility}, we can see that dishonest behavior results in utility of $b\cdot(g\cdot(-s)+(1-g)24s)+(1-b)(-s)$, which is easily simplified to $25bs-25bgs-s$.  Since we want to find the values of $b$ and $g$ for which $A$ is indifferent between honesty and dishonesty, we set this expression equal to 0, the utility of honesty.  We solve for $g$ first, readily getting,
    \begin{align}
        g=1-\frac{1}{25b} \label{g_simple}
    \end{align}
    Given that $g$ must be positive, we can read off that $b>\frac{1}{25}$. \qed 
    \end{proof}

\begin{theorem} \label{main}
    Given $s_A=s_V=s$, $x=\frac{s}{24}$, and $z=24s$, there are values of $b$, $g$, and $h$ that admit a mixed equilibrium that is not the same as that found in Lemma 2, if both $A$ and $V$ are rational and risk neutral.
\end{theorem}
\begin{proof}
From \cref{g_simple}, we can see that a value of $b=\frac{1}{5}$ yields $g=\frac{4}{5}$ and both are viable. Given these values, we can ask, for which value $h$, determined by $A$, is $V$ indifferent. We can now simply plug these values into \cref{v_utility} and set it to 0, yielding,
\begin{align*}
    \frac{1}{5}\left(\frac{4}{5}\left(h\left(-s\right)+\left(1-h\right)\frac{s}{2}\right)\right)+\frac45\left(\left(1-h\right)\frac{s}{2}-\frac{s}{24}\right)=0 
\end{align*}
Solving for $h$ yields $h= \frac{67}{96}$.  Thus, we have found values for which both $A$ and $V$ can force the other to be indifferent to their choice of behavior.  Thus we have found an equilibrium not in keeping with the desired functioning. \qed \end{proof}

\section{Potential Solutions}
The verifier's dilemma has yet to become a practical problem because existing optimistic rollups have yet to decentralize their sequencers. While this is one potential solution, it is important to ask how decentralization might be achieved without compromising the security of the chain. Another trivial solution would be to revert to a mechanism that requires the validators to check the integrity of the block every time, as is done with proof of stake L1s and zero knowledge rollups. We present here an outline of potential solutions that are more in keeping with the spirit of optimistic rollups.

\subsection{Random Checks}
If we were to stipulate that each block be checked, we would stray away from the spirit of optimistic rollups and lose their lightweight efficiency.  Requiring intermittent, random checks could keep the optimistic efficiency without the hazard of the verifier's dilemma.  Whether such a set up would indeed be meaningfully more efficient than simply checking every time would depend on the relative value of the potential benefit from attack and the potential penalties. It would also depend on precisely how the random checks would be implemented. This would not necessitate fraud proofs, which are needed to find precisely where the chain became faulty, and instead would require proofs of just the block in question.  These proofs could be done by nodes, which would again need incentives, or by smart contracts on the L1, which could have expensive execution that might fall on the users or be paid for with inflation. If we assume this could be done with a smart contract without too much additional cost to the user, we get a game where $C$ is the contract `player' in the place of $V$.  Let us also define $p$, where $0<p<1$, to be the probability with which the contract checks the proof.

\begin{lemma}
    There are values of the probability of check, $p$, $s_A$, and $z$ such that a rational, risk neutral aggregator will always be honest.
\end{lemma}

\begin{center}
\begin{istgame}[font=\footnotesize]
\xtShowEndPoints
\xtdistance{15mm}{40mm}
\istroot(0)[initial node]{A}
    \istb{Honest}[al]
    \istb{Dishonest}[ar]
\endist

\xtdistance{15mm}{20mm}
\istroot(1)(0-1)<135>{$C$}
    \istb{Check}[al]{0}
    \istb{No}[ar]{0}
\endist 

\istroot(2)(0-2)<above right>{$C$}
    \istb{Check}[al]{-s_A }
    \istb{No}[ar]{z}
\endist
\xtInfoset(1)(2)
\end{istgame}

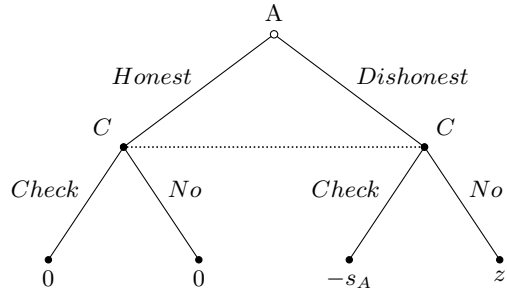
Game 3: Game tree illustrating a random check and the utility of $A$.
\end{center}
It is immediate that the expected income on the left child of the root is zero.  Recall this is because we normalized it so.  Since the left, honest choice on the part of $A$ is the desired outcome, we need a $p$ such that the right side yields negative income in expectation.  Thus we would need, $(1-p)z -ps_A<0$. Solving for $p$ readily yields $p>\frac{z}{z+s_A}$.  This is clearly in keeping with our requirement that $0<p<1$, but is also clear that, given $s_A \ll z$, $p$ will be close to $1$.  Thus, there are only minor savings as compared to checking every time, but there are values for which such a setup would result in a rational, risk neutral $A$ behaving honestly, as desired. \qed 

\subsection{Random Mistakes}
The motivating story for the verifier's dilemma is a scenario in which the aggregator makes no mistakes and thus removes the incentive to ever check.  This in turn leads to an opening for mistakes or malice on the part of the aggregator. Adding mistakes by design thus seems like a potential way to incentivize watchfulness on the part of the validators. This might pose a practical problem as presumably only some of the aggregator nodes, say those belonging to the foundation associated with the chain, would be making these purposeful mistakes and validators might then only check the blocks proposed by those nodes. Furthermore some entity would have to incur the cost of the stakes slashed when the mistakes are caught.

\subsection{Easter Eggs}
If the search costs, $x$, are relatively small, compensating the validators directly for their services seems like a tempting solution, and has been considered by Arbiturm \cite{arbitrum}.  However, it may not be clear whether or not the search is being performed correctly and checking this search would be yet another role with its own incentives. If it were possible to hide reward `Easter eggs,' which could be found through proper searching but not a random challenge, this could act as a delivery mechanism for validator rewards. These rewards would be realized even when the aggregator behaves honestly.  

\begin{lemma}
    Compensating the validator with some term with expected value $y\geq x$ will result in the expected income for the validator of always searching being higher than not searching.
\end{lemma}
Let us add $y$, the expected value of finding the Easter egg, to the validator income in the case of search, i.e. when $b=0$, which we can read off from \cref{v_utility}.  This should yield greater expected utility than never searching, i.e. when $b=1$.  
\begin{align*}
    (1-h)\frac12 s_A -x+y >-ghs_V +g(1-h)\frac{1}{2}
\end{align*}
since $g<1$, it is clear to see that if we have $-x+y >-ghs_V$, the above term is always true. If we require that $y \geq x$, the result is immediate. \qed











\printbibliography

\section*{Appendix}

\begin{lemma}
    A rational, risk neutral $A$ will be indifferent to being honest if $g= 1- \frac{s_A}{bs_A +bz}$.
\end{lemma}
\begin{proof}
    It is easy to read off from Game 2 that if the aggregator is honest, they will receive $0$ utility, regardless of $V$'s actions.  Recall this is due to our normalization of $f-w$ to $0$. 
    
    Suppose temporarily that $A$ knew that $V$ was not doing searches and was challenging randomly with a probability $g$.  If they were then dishonest, they would then have utility of $g(-s_A)+(1-g)z$.
    Let $b$ be the probability that $V$ \emph{blindly} challenges, i.e. incurs no search costs. Honest behavior still results in $0$ and dishonest behavior results in expected utility of 

    $$
        b\cdot(g\cdot(-s_A)+(1-g)z)+(1-b)(-s_A).
    $$
    This is the sum, weighted by $b$, for $V$ going in blind, selecting challenge with probability $g$, and by $(1-b)$ for $V$ performing a search and catching $A$ with certainty.

    If we set this equal to $0$, which is the expected utility to $A$ of behaving honestly, we get 
    \begin{align*}
        b\cdot(g\cdot(-s_A)+(1-g)z)+(1-b)(-s_A)=0 \\
        -bgs_A + bz -bgz -s_A +bs_A = 0 
    \end{align*}
    Solving for $g$ yields
    \begin{align*}
        bz-s_A+bs_A =  g\cdot(bs_A +bz)\\
        g = \frac{bz-s_A+bs_A}{bs_A +bz} 
    \end{align*}
We can rearrange this slightly differently.
    \begin{align}
        g= \frac{bz+bs_A}{bs_A +bz} -\frac{s_A}{bs_A +bz} \nonumber \\
        g= 1- \frac{s_A}{bs_A +bz} \label{g_bound}
    \end{align}
This is a line that traces out the values of $g$ and $b$ for which the aggregator is indifferent between being honest or dishonest.

In order for this to be a viable scenario, we must have that $0<b<1$ and $0<g<1$.  This requirement will clearly be broken if \cref{g_bound} is less than $0$.  Thus, we have a viable $g$ if,
\begin{align*}
    0 < 1- \frac{s_A}{bs_A +bz} \\
    \frac{s_A}{bs_A +bz} <  1 \\
    \frac{s_A}{s_A +z} < b
\end{align*}
It is easy to see that this can be satisfied with a $b$ between $0$ and $1$.  Given that $z \gg s_A$, this bound is easily achieved in context. \qed \end{proof}
\begin{lemma} \label{h bound}
    There is a value for $h$, the probability with which the aggregator is honest, for which the validator is indifferent between search and no search.
\end{lemma}

\begin{proof}
    Referring again to Game 2 and starting on the left side, where $V$ chooses no search, we can see that the validator reaps $0$ utility if they do not challenge. In the challenge case, if $A$ is honest, $V$ has their stake burned for utility $-s_V$.  In the dishonest case, they receive $\frac12 s_A$.  Weighting this with the probabilities, we find the left hand side has expected utility 
    $$
        g\cdot\left(h\cdot(-s_V) + (1-h)\frac12 s_A\right).
    $$
    On the right hand side, $V$ can now make an informed decision.  If $A$ is honest, $V$ will not challenge, resulting in utility $-x$ and if $A$ is dishonest, $V$ will always challenge for utility $\frac12 - x$, which represents the bounty less the search cost.  Weighting these terms by the honesty probability yields the right hand side utility of 
    $$
        h\cdot(-x)+ (1-h)\left(\frac12 s_A -x\right)
    $$
    which can be simplified to
    $$
        (1-h)\left(\frac12 s_A \right)-x    
    $$

    To get the whole expected utility, we weight both sides of the tree by the probability of no search $b$, which gives,
    $$
         bg\cdot\left(h\cdot(-s_V) + (1-h)\frac12 s_A\right)+(1-b)\left((1-h)\frac12 s_A -x\right).
    $$
We then set this to $0$ to find the $h$ for which the validator is indifferent. 
\begin{align}
     bg\cdot\left(h\cdot(-s_V) + (1-h)\frac12 s_A\right)+(1-b)\left((1-h)\frac12 s_A -x\right)=0 \nonumber \\
     -bghs_V+bg\frac12 s_A- bgh\frac12 s_A + (1-b)\frac12 s_A - (1-b)h\frac12 s_A-(1-b)x=0 \nonumber \\
     bg\frac12 s_A+ (1-b)\frac12 s_A-(1-b)x=h\cdot \left(bgs_V + bg\frac12 s_A  + (1-b)\frac12 s_A\right) \nonumber \\
     h= \frac{bg\frac12 s_A+ (1-b)\frac12 s_A-(1-b)x}{bgs_V + bg\frac12 s_A  + (1-b)\frac12 s_A} \label{h_bound}
\end{align}
Which gives us a condition on $h$ as desired.  \qed
\end{proof}

We can combine \cref{g_bound} and \cref{h_bound} to find values for which both conditions hold.  Since we have three variables but only two equations, we will not have a single point for our equilibrium.

\begin{theorem} \label{main}
    There are values of $s_A$, $s_V$, $z$, and $x$ that admit a mixed equilibrium that is not the same as that found in Lemma 2 if both $A$ and $V$ are rational and risk neutral.
\end{theorem}
\begin{proof}
    We can plug  \cref{g_bound} into \cref{h_bound} to attain a simplified closed form solution,
    \begin{align}
           h=\frac{\frac12 s_A z- (1-b)(zx+s_Ax)}{\frac12 s_Az-s_A s_V +(z+s_A)s_Vb}, \label{combined_bound}
    \end{align}

after some tedious manipulations that we will omit.  In context, it is reasonable to assume that the search costs are relatively small. If one stake is larger it would be that of the aggregator. The aggregator can extract a large amount of value, $z$.  This yields the ordinal constraint:
$$
    x<s_V\leq s_A \ll z
$$
We must also have that $0<h<1$ and $0<b<1$.  We can think of the right hand side of \cref{combined_bound} as a function in terms of $b$. In order for their to be a mixed equilibrium that is not a pure strategy, we need there to be a viable value of $h(b)$, i.e. between $0$ and $1$, for some value $b$ where $0<b<1$. In order for that to be the case, we must have that the denominator is larger than the numerator and for both to have the same sign. Starting with the numerator of \cref{combined_bound}, we can see that it will be positive when
\begin{align*}
    \frac12 s_A z-zx-s_Ax + (zx+s_Ax)b > 0 \\
    \left(\frac12 s_A -x  + xb \right)z -s_Ax + s_Axb > 0
\end{align*}

We can see that if the term $\frac12 s_A -x  + xb $ is positive, then, given that $z$ is \emph{very large}, this inequality will be satisfied.  It is easy to see that this term will be positive if 
\begin{align*}
    \frac12 s_A -x  + xb> 0 \\
    \frac12 s_A > (1-b)x
\end{align*}
This is a reasonable restraint given that the stakes is considerably larger than the search cost, $x$.

Turning our attention to the denominator, it is easy to read off that the term will be positive if $\frac12 z >s_V$.  This will be the case again given that $z$ \emph{very large}.

We will now have that the whole term for $h$ is between $0$ and $1$ if the numerator is less than the denominator. That is, 
\begin{align*}
    \frac12 s_A z- (1-b)(zx+s_Ax) < \frac12 zs_A-s_A s_V +(z+s_A)s_Vb 
\end{align*}
which can be simplified to 
\begin{align*}
    (b-1)(zx+s_Ax) < zs_Vb+(b-1)s_As_V \\
    bzx-zx+bs_Ax+s_Ax < zs_Vb+bs_As_V - s_As_V
\end{align*}
Once again solving for $b$ yields
\begin{align}
    s_As_V-s_Ax-xz<&b(zs_v+s_As_V-xz-s_Ax) \nonumber \\
    \frac{s_As_V-s_Ax-zx}{s_As_V-s_Ax-zx +zs_V}<&b \label{last_b_bound}
\end{align}
We see that this constraint on be conforms to the requirement that $0<b<1$. Thus there are values for $g$, $h$, and $b$ all strictly between $0$ and $1$ that admit an equilibrium. This is not the pure strategy equilibrium from Lemma 2 that is intended by the optimistic design. \qed \end{proof}

\subsection{The Role of the Transactor}
\begin{lemma}
    A rational, risk neutral transactor will still participate if their potential utility $u_T$ is sufficiently high.
\end{lemma}
\begin{proof}
  We must have that the transactor $T$ is willing to participate rather than take the outside option.  Their utility when $A$ is honest is $u_T -f$.  If $A$ is dishonest, but is caught, $T$ receives $0$ utility.  If $A$ is dishonest, but is not caught, which happens if $V$ is blind and does not challenge, $T$ receives $-f$ utility because they lose their fee.  Combining this with the probability weights yields $h(u_T -f)-b(1-h)(1-g)f$.  When this term is greater than $0$, $T$ will still participate.  We can use this to find a condition on $u_T$
\begin{align*}
    h\cdot(u_T -f)-b\cdot(1-h)(1-g)f > 0 \\
    u_T > f+ \frac{b\cdot(1-g)(1-h)}{h}f
\end{align*}
From this, we can see that as long as $u_t$ is somewhat larger than $f$, and as long as $h$ is not \emph{too small}, the transactor will still participate.    \qed
\end{proof}

What we have shown is that strategies other than those we found in Lemma 2 constitute the equilibrium.  We have found constraints on the model values such that if those constraints are met, the validator does no search $b$ of the time and, given no search, blindly challenges $g$ of the time. The validator is honest only $h$ of the time. The transactor will still participate despite this added uncertainty. The precise values of these depend on the values of the search cost, the stakes, and the potential benefit of getting away with dishonest behavior.

\subsection{Remark}
With \cref{main}, we have shown a result that has been described before, but we believe our model to be a more thorough articulation of the functioning of optimistic rollups.  Our contribution further bolsters the understanding that the verifier's dilemma poses a significant problem to the well-functioning of optimistic rollups.  Of greater consequence, our work also illuminates the roles played by each factor in this problem.  For example, if it were not for the validator stake $s_V$, or if the validator stake were very small, the lower bound on the probability of no search, $b$, would be $1$ or very close to it, as can be seen in \cref{last_b_bound}.  This would mean that the validator would rather incur the penalty and check randomly than incur the search cost and check with information.  Random checking  cannot necessarily be written off as a potential solution, as we will discuss in the next section.  The precise value of $x$ is also important to the existence and severity of the deviation from the expected functionality. This again points towards a potential solution of alleviating this cost for validators, which we will also discuss in the next section. Furthermore, we can see that there are values for which the equilibrium we have found does not necessarily work. Further study of these values could yield a solution within the existing framework or better illuminate the potential risk of decentralization.  All of these observations underscore the importance of the model we present here, which is, to the best of our knowledge, the most accurate and nuanced capture of the optimistic rollup mechanism.

\end{document}